\RequirePackage{fix-cm}
\RequirePackage{amsmath}
\documentclass[a4paper,11pt]{article} 

\usepackage{graphicx}
\usepackage{times}
\usepackage{amsmath,amssymb,amsopn,amsthm}
\usepackage{url}
\usepackage[caption=false]{subfig}
\usepackage{tikz}
\usepackage{fullpage}

\newtheorem{definition}{Definition}
\newtheorem{theorem}{Theorem}
\newtheorem{lemma}{Lemma}

\newtheorem{corollary}{Corollary}
\theoremstyle{plain}
\newtheorem{clm}{Claim}{\it}{\rm}

\DeclareMathOperator*{\argmax}{arg\,max}

\newcommand{\Exp}[1]{\mathbb{E}\left[\,#1\,\right]}
\renewcommand{\Pr}[1]{\mathbf{Pr}\left[\;#1\;\right]}

\newcommand{\x}{\mathbf{x}}
\newcommand{\y}{\mathbf{y}}
\newcommand{\z}{\mathbf{z}}
\newcommand{\s}{\mathbf{s}}
\newcommand{\ttt}{\mathbf{t}}
\newcommand{\q}{\mathbf{q}}
\newcommand{\rr}{\mathbf{r}}

\usetikzlibrary{arrows.meta,patterns}
\usetikzlibrary{shapes.geometric}
\usetikzlibrary{decorations.text}

\tikzset{global scale/.style={
		scale=0.6,
		every node/.style={scale=0.6}
	}
}

\newcommand{\remove}[1]{}
\usepackage[hidelinks]{hyperref}

\begin{document}

\title{Impartial Selection with Additive Approximation Guarantees\thanks{
	A preliminary version \cite{caragiannis2019impartial} of this paper was presented in the 12th International Symposium on Algorithmic Game Theory (SAGT '19).}
}

\author{
	Ioannis Caragiannis\thanks{ Department of Computer Science,
		Aarhus University, \textit{iannis@cs.au.dk}} \and
	George Christodoulou\thanks{ Department of Computer Science,
		Aristotle University of Thessaloniki, \textit{ gichristo@csd.auth.gr }} \and
	Nicos Protopapas\thanks{Department of Computer Science,
		University of Liverpool, \textit{N.Protopapas@liverpool.ac.uk}}
}
\date{}

\maketitle

\begin{abstract}
 Impartial selection has recently received much attention within the
multi-agent systems community. The task is, given a directed graph
representing nominations to the members of a community by other
members, to select a member with the highest number of
nominations. This seemingly trivial goal becomes challenging when
there is an additional impartiality constraint, requiring that no
single member can influence her chance of being selected. Recent
progress has identified impartial selection rules with optimal
approximation ratios. Moreover, it was noted that worst-case
instances are graphs with few vertices. Motivated by this fact, we
propose the study of {\em additive approximation}, the difference between
the highest number of nominations and the number of nominations of the selected member,
as an alternative measure of the quality of impartial selection.

Our positive results include two randomized impartial selection
mechanisms which have additive approximation guarantees of
$\Theta(\sqrt{n})$ and $\Theta(n^{2/3}\ln^{1/3}n)$ for the two most
studied models in the literature, where $n$ denotes the community size. 
We complement our positive results by providing negative results for
various cases. First, we provide a characterization for the
interesting class of strong sample mechanisms, which allows us to
obtain lower bounds of $n-2$, and of $\Omega(\sqrt{n})$ for their
deterministic and randomized variants respectively. Finally, we present
a general lower bound of $3$ for all deterministic impartial
mechanisms.
\end{abstract}

\section{Introduction}\label{sec:introduction}

We study the problem that arises in a community of individuals that
want to select a community member to receive an award. This is
a standard social choice \cite{comsoc-book} problem, that is typically
encountered in scientific and sports communities but has also found
important applications in distributed multi-agent systems. To give an
entertaining example, the award for the player of the
year\footnote{Some basic information for the award can be found in \url{https://en.wikipedia.org/wiki/PFA_Players\%27_Player_of_the_Year}}
by the Professional Footballers Association (PFA) is decided by the
members of PFA themselves; each PFA member votes for the two players they
consider the best for the award and the player with the maximum
number of votes receives the award. Footballers consider it as one of the
most prestigious awards, due to the fact that it is decided by their
opponents. In distributed multi-agent systems, {\em leader election}
(e.g., see \cite{attiya-book}) can be thought of as a selection
problem of similar flavor.  
Other notable examples include (see \cite{fischer2015}) the selection of a representative in a group,	funding decisions based on peer reviewing or even (see \cite{alon11}) finding the most popular user of a social network.

The input of the problem can be represented as a directed graph, which
we usually call nomination profile. Each vertex represents an
individual and a directed edge indicates a vote (or nomination) by a
community member to another. 
A {\em selection mechanism} (or {\em selection rule}) takes a
nomination profile as input and returns a single vertex as the
winner. Clearly, there is a highly desirable selection rule: the one
which always returns the highest in-degree vertex as the
winner. Unfortunately, such a rule suffers from a drawback that is
pervasive in social choice. Namely, it is {\em susceptible to manipulation}.

In particular, the important constraint that makes the selection
challenging is {\em impartiality}. As every individual has a personal
interest to receive the award, selection rules should take the
individual votes into account but in such a way that no single
individual can increase her chance of winning by changing her
vote. The problem, known as {\em impartial selection}, was introduced
independently by Holzman and Moulin~\cite{moulin13} and Alon et
al.~\cite{alon11}. Unfortunately, the ideal selection rule mentioned
above is not impartial. Consider the case with a few individuals that
are tied with the highest number of votes. The agents 
involved in the tie might be tempted to lie about their true 
preferences to break the tie in their favor.\footnote{ As an illustrative example, consider a rule which assigns the prize to the maximum in-degree vertex, and if a tie exists, is resolved by some arbitrary tie breaking rule. It is not hard to see that any such rule cannot be impartial. Indeed, consider the case where two vertices, $a$ and $b$ have maximum in-degree, they both vote for each other and let $a$ be the winner according to the tie breaking rule. Observe that $b$ has incentive to decrease $a$'s in-degree, by removing its outgoing edge towards $a$, and become the sole maximum in-degree vertex, and thus the winner.}  

Impartial selection rules may inevitably select as the winner a vertex
that does not have the maximum in-degree. Holzman and Moulin~\cite{moulin13} considered minimum axiomatic properties that
impartial selection rules should satisfy. For example, a highly
desirable property, called negative unanimity, requires that an
individual with no votes at all, should never be selected. Alon et
al.~\cite{alon11} quantified the efficiency loss with the notion of
approximation ratio, defined as the worst-case ratio of the maximum
vertex in-degree over the in-degree of the vertex which is selected by
the rule. According to their definition, an impartial selection rule
should have as low approximation ratio as possible. This line of
research was concluded by the work of Fischer and
Klimm~\cite{fischer2015} who proposed impartial mechanisms
with the optimal approximation ratio of $2$.

It was pointed out in \cite{alon11,fischer2015}, that the most
challenging nomination profiles for both deterministic and randomized
mechanisms are those with small in-degrees. In the case of
deterministic mechanisms, the situation is quite extreme as all
deterministic mechanisms can be easily seen to have an unbounded
approximation ratio on inputs with a maximum in-degree of $1$ for a
single vertex and $0$ for all others; see~\cite{alon11} for a concrete
example. As a result, the approximation ratio does not seem to be an
appropriate measure to classify deterministic selection mechanisms.
Finally, Bousquet et al.~\cite{bousquet2014} have shown that if the
maximum in-degree is large enough, randomized mechanisms that return a
near optimal impartial winner do exist.

We deviate from previous work and instead propose to use {\em additive
approximation} as a measure of the quality of impartial selection
rules. Additive approximation is defined using the {\em difference}
between the maximum in-degree and the in-degree of the winner returned
by the selection mechanism. Note that deterministic mechanisms with
low additive approximation always return the highest in-degree vertex
as the winner when her margin of victory is large. When this does
not happen, we have a guarantee that the winner returned by the
mechanism has a close-to-maximum in-degree. 

\smallskip 

\noindent {\em Our contribution.} \label{subsec:contribution}
We provide positive and negative results for
impartial selection mechanisms with additive approximation
guarantees. We distinguish between two models. In the first model,
which was considered by Holzman and Moulin~\cite{moulin13}, nomination
profiles consist only of graphs with all vertices having an out-degree
of $1$. The second model is more general and allows for multiple
nominations and abstentions (hence, vertices have arbitrary
out-degrees).

As positive results, we present two randomized impartial mechanisms
which have additive approximation guarantees of $\Theta(\sqrt{n})$ and
$\Theta(n^{2/3}\ln^{1/3}n)$ for the single nomination and multiple
nomination models, respectively. Notice that both these additive
guarantees are $o(n)$ functions of the number $n$ of vertices. We
remark that an $o(n)$-additive approximation guarantee can be
translated to an $1-\epsilon$ multiplicative guarantee for graphs with
sufficiently large maximum in-degree, similar to the results of
\cite{bousquet2014}. Conversely, the multiplicative guarantees of
\cite{bousquet2014} can be translated to an $O(n^{8/9})$-additive
guarantee\footnote{The authors in \cite{bousquet2014} do not provide
	additive guarantees, hence we based our calculations on their
	provided bounds on the multiplicative guarantee $1-\epsilon$. It is
	important to note however that they claim that they have
	not optimized their parameters, so it is possible that this guarantee
	can be further reduced by a tighter analysis.}.  This analysis
further demonstrates that additive guarantees allow for a more
smooth classification of mechanisms that achieve good multiplicative
approximation in the limit.

Our mechanisms first select a small sample of vertices, and then
select the winner among the vertices that are nominated by the sample
vertices. These mechanisms are randomized variants of a class of
mechanisms which we define and call {\em strong sample
mechanisms}. Strong sample mechanisms are impartial
mechanisms which select the winner among the vertices nominated by a
sample set of vertices. In addition, they have the characteristic that
the sample set does not change with changes in the nominations of the
vertices belonging to it. For the single nomination model, we provide
a characterization, and we show that all deterministic strong sample mechanisms should use a
fixed sample set that does not depend on the nomination profile. This
yields a $n-2$ lower bound on the additive approximation guarantee of
any deterministic strong sample mechanism. For their randomized variants,
where the sample set is selected randomly, we present an
$\Omega(\sqrt{n})$ lower bound which shows that our first
randomized impartial mechanism is best possible among all randomized
variants of strong sample mechanisms. Finally, for the most general,
multiple nomination model, we present a lower bound of $3$ for all
deterministic mechanisms.

\smallskip

\noindent {\em Related work.} 
Besides the papers by Holzman and Moulin~\cite{moulin13} and Alon et
al.~\cite{alon11}, which introduced impartial selection as we study it
here, de Clippel et al.~\cite{declipper2008} considered a different
version with a divisible award. Alon et al.~\cite{alon11} used the
approximation ratio as a measure of quality for impartial selection
mechanisms. After realizing that no deterministic mechanism achieves a
bounded approximation ratio, they focused on randomized mechanisms and
proposed the \textsc{$2$-Partition} mechanism, which guarantees an
approximation ratio of $4$ and complemented this positive result with a
lower bound of $2$ for randomized mechanisms.

Later, Fischer and Klimm were able to design a mechanism that achieves an approximation ratio of $2$, by generalizing \textsc{$2$-partition}. Their optimal mechanism, called \textsc{Permutation}, examines the vertices sequentially following their order in a random permutation and selects as the winner the vertex of highest degree counting only edges with direction from ``left'' to ``right.'' They also provided lower bounds on the approximation ratio for restricted inputs (e.g., with no abstentions) and have shown that the worst case examples for the approximation ratio are tight when the input nomination profiles are small.

Bousquet et al.~\cite{bousquet2014} noticed this bias towards
instances with small in-degrees and examined the problem for instances
of very high maximum in-degree. After showing that
\textsc{Permutation} performs significantly better for instances of
high in-degree, they have designed the \textsc{Slicing} mechanism with
near optimal asymptotic behaviour for that restricted family of
graphs. More precisely, they have shown that, if the maximum in-degree
is large enough, \textsc{Slicing} can guarantee that the winner's
in-degree approximates the maximum in-degree by a small error. As we
discussed in the previous section, the \textsc{Slicing} mechanism can
achieve an additive guarantee of $O(n^{8/9})$.

Holzman and Moulin~\cite{moulin13} explored impartial mechanisms through an axiomatic approach. They focused on the single nomination model and proposed several deterministic
mechanisms, including the \textsc{Majority with Default}
mechanism. \textsc{Majority with Default} defines a vertex as a default winner
and examines if there is any vertex with in-degree more than $\lceil
n/2\rceil$, ignoring the outgoing edge from the default winner. If such a vertex
exists, then this is the winner; otherwise the default vertex
wins. While this mechanism has the unpleasant property that the
default vertex may become the winner with no incoming edges at all, its additive
approximation is at most $\lceil n/2\rceil$.  Further to that, they
came up with a fundamental limitation of the problem: no impartial
selection mechanism can be simultaneously negative and positive unanimous (i.e., never selecting as a winner a vertex of in-degree $0$ and always selecting the vertex of in-degree $n-1$, whenever there exists one).

Mackenzie in \cite{mackenzie2015} characterized symmetric (i.e., name-independent) rules in the single nomination model. Tamura and Ohseto \cite{tamura2014impartial} observed that when the demand for only one winner is relaxed, then impartial, negative unanimous and positive unanimous mechanisms do exist. Later on, Tamura \cite{tamura2016characterizing} characterized them. On the same agenda, Bjelde et al. in \cite{bjelde2017} proposed a deterministic version of the permutation mechanism that achieves the $1/2$ bound by allowing at most two winners. Alon et al.~\cite{alon11} also present results for selecting multiple winners. 

Finally, we remark that impartiality has been investigated as a desired property in other contexts where strategic behaviour occurs. Recent examples include peer reviewing~\cite{aziz2016strategyproof,kahng2018ranking,Kurokawa2015}, selecting impartially the most influential vertex in a network~\cite{babichenko2018incentive} and in linear regression algorithms as a means to tackle strategic noise \cite{chen2018}.

\section{Preliminaries}\label{sec:preliminaries}

Let $N=\{1,...,n\}$ be the set of $n\geq 2$ agents. A \emph{nomination graph} $G=(N,E)$ is a directed graph with vertices representing the agents. The set of outgoing edges from each vertex represents the
nominations of each agent; it contains no self-loops (as, agents are not allowed to nominate themselves) and can be empty (as an agent is, in general, allowed to abstain). We
write $\mathcal{G}=\mathcal{G}_n$ for the set of all graphs with $n$
vertices and no self-loops. We also use the notation
$\mathcal{G}^1=\mathcal{G}^1_n$ to denote the subset of $\mathcal{G}$
with out-degree exactly $1$. For convenience in the proofs, we sometimes
denote each graph $G$ by a tuple $\x$, called \emph{nomination profile}, where $x_{u}$ denotes the set of outgoing edges of vertex $u$ in $G$.  For $u \in N$, we use the notation
$\x_{-u}$ to denote the graph $(N,E\setminus (\{u\} \times
N))$ and, for the set of vertices $U\subseteq N$, we use $\x_{-U}$ to denote the graph
$(N,E\setminus (U \times N))$. We use the terms nomination graphs and nomination profiles interchangeably.

The notation $\delta_{S}(u,\x)$ refers to the in-degree of
vertex $u$ in the graph $\x$ taking into account only edges that
originate from the subset $S\subseteq N$. When $S=N$, we use the
shorthand $\delta(u,\x)$ and when $S=\{v\}$ we use the simplified notation $\delta_v(u,\x)$.
If the graph is clearly identified by the
context we omit $\x$ too, using $\delta(u)$. We denote the maximum in-degree of graph $\x$ as $\Delta(\x)= \max_{u \in N}
\delta(u,\x)$ and, whenever $\x$ is clear from the context, we use
$\Delta$ instead.

A \emph{selection mechanism} for a set of graphs $\mathcal{G}' \subseteq
\mathcal{G}$, is a function $f: \mathcal{G'} \rightarrow [0,1]^{n+1}$,
mapping each graph of $\mathcal{G}'$ to a probability distribution
over all vertices (which can be potential winners) as well as to the possibility of returning no winner at all. A selection mechanism is {\em deterministic} in the
special case where for all $\x$, $(f(\x))_u\in \{0,1\}$ for all vertices $u \in N$.

A selection mechanism is \emph{impartial} if for all graphs $\x \in
\mathcal{G}'$, all possible sets $x'_u$ of outgoing edges (from vertex $u$),  it holds $(f(\x))_u=(f(x'_u,\x_{-u}))_u$ for every vertex $u$. In words,
the probability that $u$ wins must be independent of the set of its
outgoing edges.

Let  $\Exp{\delta(f(\x))}$ denote the expected in-degree of $f$ on $\x$, i.e. $\Exp{\delta (f(\x)) }=\sum_{u \in N} (f(\x))_u \delta(u,\x) $. A selection mechanism $f$ is called  $\alpha(n)$-additive if $$\max_{\x \in \mathcal{G}_n} \left\{\Delta(\x) - \Exp{\delta(f(\x))}\right\} \leq \alpha(n),$$ for every $n\in \mathbb{N}$.

\section{Upper Bounds}\label{sec: upperbounds}
In this section we provide randomized selection mechanisms for the two
best studied models in the literature. First, in
Section~\ref{sec:random-upper} we propose a mechanism for the single nomination model of
Holzman and Moulin~\cite{moulin13}, where nomination profiles consist
only of graphs with all vertices having an out-degree of $1$. Then, in
Section~\ref{sec:general-upper-bound} we provide a mechanism for the
more general model studied by Alon et al.~\cite{alon11}, which allows
for multiple nominations and abstentions.

\subsection{The {\sc Sample~and~Vote} Mechanism}\label{sec:random-upper}
Our first mechanism, {\sc Sample~and~Vote}, forms a sample $S$ of
vertices by repeating $k$ times the selection of a vertex uniformly at
random with replacement\footnote{Sampling uniformly at random with replacement allows for a simple analysis of the mechanism. In section \ref{sec:lower-bound}, we show that this is indeed a good choice, as no other sampling method yields better additive approximation.}. Any vertex that is selected at least once
belongs to the sample $S$. Let $W:=\{u\in N\setminus
S:\delta_S(u,\x)\geq 1\}$ be the set of vertices outside $S$ that are
nominated by the vertices of $S$. If $W=\emptyset$, no winner is
returned. Otherwise, the winner is a vertex in $\argmax_{u\in
	W}{\delta_{N\setminus W}(u,\x)}$. We note here the crucial fact that the selection of the sample set $S$ is independent of the profile $\x$.

Impartiality follows since a vertex that does not belong to $W$ (no
matter if it belongs to $S$ or not) cannot become the winner and the
nominations of vertices in $W$ are not taken into account for deciding
the winner among them. We now argue that, for a carefully selected $k$,
this mechanism also achieves a good additive guarantee.

\begin{figure*}[ht]
	\centering
	\subfloat[{\sc Sample~and~Vote}]{\label{fig:UB:SampleAndVote}
		\begin{tikzpicture}[winner/.style={circle,draw=black!80,fill=black!80!,very thick,minimum size=0.8cm,text=white},simple/.style={very thick,circle,draw=black!80,minimum size=0.8cm},sample/.style={circle, very thick,draw=black!80!red,fill=black!30!,minimum size=0.8cm, dashed},cand/.style={circle,draw=red!80,fill=black!10,minimum size=0.8cm,very thick}]
		
		\fill[black!20!white,rotate=45] (1,0) ellipse (2 cm and 0.6 cm);
		\node[simple] at (0,0) (c) {$5$};
		\node[winner] at  (45:2) (a1) {$4$};
		\node[simple] at  (75:3.5) (aa1) {$6$};
		\node[simple] at  (55:3.5) (aa2) {$7$};
		\node[simple] at  (35:3.5) (aa3) {$8$};
		\node[sample] at  (90:2) (a2) {$3$};
		\node[sample] at  (90+45:2) (a3) {$2$};
		\node[simple] at  (180:2) (a4) {$1$};
				
		\draw[-Latex,-Latex,thick ] (aa3) to (a1);
		\draw[-Latex,-Latex,thick ] (aa2) to (a1);
		\draw[-Latex,-Latex,thick ] (aa1) to (a1);
		\draw[-Latex,dotted ] (a1) to (c);
		\draw[-Latex,,-Latex,thick ] (a2) to (c);
		\draw[-Latex, ,-Latex,thick] (a3) to (c);
		\draw[-Latex, bend left,,-Latex,thick] (a4) to (c);
		\draw[-Latex , bend left,dotted] (c) to (a4);
		
		\begin{scope}[xshift=100]
		\node[sample] at  (0:0) (w) {$12$}; 
		\node[simple] at  (90-45:-2) (y) {$10$}; 
		\node[simple] at  (0:-2) (x) {$9$}; 
		\node[simple] at  (90:-2) (z) {$11$}; 
		
		\draw[-Latex,dotted] (x) to (w);
		\draw[-Latex,dotted ] (y) to (w);
		\draw[-Latex,dotted ] (z) to (w);
		\draw[-Latex,-Latex,thick ] (w) to (a1);
		\end{scope}
		\end{tikzpicture}
	}\hfill
	\subfloat[{\sc Sample~and~Poll}]{\label{fig:UB:SampleAndPoll}
		\centering
		\begin{tikzpicture}[winner/.style={circle,draw=black!80,fill=black!80!,very thick,minimum size=0.8cm,text=white},simple/.style={very thick,circle,draw=black!80,minimum size=0.8cm},sample/.style={circle, very thick,draw=black!80!red,fill=black!30!,minimum size=0.8cm, dashed},cand/.style={circle,draw=red!80,fill=black!10,minimum size=0.8cm,very thick}]
		
		\node[winner] at (0,0) (c) {$5$};
		\node[simple] at  (45:2) (a1) {$4$};
		\node[simple] at  (75:3.5) (aa1) {$6$};
		\node[simple] at  (55:3.5) (aa2) {$7$};
		\node[simple] at  (35:3.5) (aa3) {$8$};
		\node[sample] at  (90:2) (a2) {$3$};
		\node[sample] at  (90+45:2) (a3) {$2$};
		\node[simple] at  (180:2) (a4) {$1$};
				
		\draw[-Latex,-Latex,dotted ] (aa3) to (a1);	
		\draw[-Latex,-Latex,dotted ] (aa2) to (a1);
		\draw[-Latex,-Latex,dotted ] (aa1) to (a1);
		\draw[-Latex,dotted ] (a1) to (c);
		\draw[-Latex,,-Latex,thick ] (a2) to (c);
		\draw[-Latex, ,-Latex,thick] (a3) to (c);
		\draw[-Latex, bend left,-Latex,dotted] (a4) to (c);
		\draw[-Latex , bend left,dotted] (c) to (a4);
		
		\begin{scope}[xshift=100]
		\node[sample] at  (0:0) (w) {$12$}; 
		\node[simple] at  (90-45:-2) (y) {$10$}; 
		\node[simple] at  (0:-2) (x) {$9$}; 
		\node[simple] at  (90:-2) (z) {$11$}; 
		
		\draw[-Latex,dotted] (x) to (w);
		\draw[-Latex,dotted ] (y) to (w);
		\draw[-Latex,dotted ] (z) to (w);
		\draw[-Latex,-Latex,thick ] (w) to (a1);
		\end{scope}
		\end{tikzpicture}
	}
	\caption{Examples for {\sc Sample~and~Vote} and {\sc Sample~and~Poll}, with sample size $k=3$ and $n=12$. In both cases we use the same sample set $S=\{2,3,12\}$. For {\sc Sample and Vote}, the vertices in $S$ define the set $W=\{4,5\}$ of possible winners. The winner is then the vertex with maximum in-degree from the votes from $N\setminus W$ to $W$ (the solid drawn edges in the figure).  
	For {\sc Sample~and~Poll}, the sample set $S$ immediately declares the winner, as one of the maximum in-degree vertices from edges starting in $S$, while the edges from vertices in $N\setminus S$ are completely ignored. In both cases, the dark vertex is the winner and the light dashed-lined vertices belong to the sample set $S$. Also, all edges drawn with a dotted line are ignored by the mechanism. The shaded area in figure \ref{fig:UB:SampleAndVote} shows which vertices belong in set $W$.
	}
	\label{fig:SampleExamples}
\end{figure*}

\begin{theorem} 
	For $k=\Theta(\sqrt{n})$, the {\sc Sample~and~Vote} mechanism is impartial and $\Theta(\sqrt{n})$-additive in the single nomination model.
\end{theorem}

\begin{proof}
	Consider a nomination graph and let $u^*$ be a vertex of maximum
	in-degree $\Delta$. In our proof of the approximation guarantee, we
	will use the following two technical lemmas.
	\begin{lemma}\label{claim:maxdegree-k}
		If $u^*\in W$, then the winner has in-degree at least $\Delta-k$.
	\end{lemma}
	
	\begin{proof}
		This is clearly true if the winner returned by {\sc Sample~and~Vote} is $u^*$. Otherwise, the winner $w$ satisfies
		\begin{align*}
		\delta(w,\x) &\geq \delta_{N\setminus W}(w,\x) \geq \delta_{N\setminus W}(u^*,\x)=\delta(u^*,\x)-\delta_W(u^*,\x)\geq \Delta-k.
		\end{align*}

		The first inequality is trivial. The second inequality follows by the
		definition of the winner $w$. The third inequality follows since $W$
		is created by nominations of vertices in $S$, taking into account that
		each vertex has out-degree exactly 1. Hence, $\delta_W(u^*,\x)\leq
		|W|\leq |S|\leq k$.
	\end{proof}
	
	\begin{lemma}\label{claim:ustar-in-W}
		The probability that $u^*$ belongs to the nominated set $W$ is 
		\[\Pr{u^*\in W} = \left(1-\left(1-\frac{\Delta}{n-1}\right)^k\right)\left(1-\frac{1}{n}\right)^k.\]
	\end{lemma}
	
	\begin{proof}
		Indeed, $u^*$ belongs to $W$ if it does not belong to the sample $S$
		and instead some of the $\Delta$ vertices that nominate $u^*$ is
		picked in some of the $k$ vertex selections. The probability that
		$u^*$ is not in the sample is
		\begin{align}\label{eq:prob-ustar-not-in-S}
		\Pr{u^*\not\in S} &=\left(1-\frac{1}{n}\right)^k,
		\end{align}
		i.e., the probability that vertex $u^*$ is not picked in some of the
		$k$ vertex selections. Observe that the probability that some of the
		$\Delta$ vertices that nominate $u^*$ is picked in a vertex selection
		step assuming that $u^*$ is never selected is
		$\frac{\Delta}{n-1}$. Hence, the probability that some of the $\Delta$
		vertices nominating $u^*$ is in the sample assuming that $u^*\not\in
		S$ is
		\begin{align}\label{eq:ustar-in-W-conditional}
		\Pr{\delta_S(u^*,\x)\geq 1|u^*\not\in S} & =1-\left(1-\frac{\Delta}{n-1}\right)^k.
		\end{align}
		The lemma follows by the chain rule
		\begin{align*}
		\Pr{u^*\in W}&=\Pr{u^*\not=S \land \delta_S(u^*,\x)\geq 1}\\
		&=\Pr{\delta_S(u^*,\x)\geq 1|u^*\not\in S}\cdot \Pr{u^*\not\in S}
		\end{align*}
		and equations (\ref{eq:prob-ustar-not-in-S}) and (\ref{eq:ustar-in-W-conditional}). 
	\end{proof}
	
	By Lemmas~\ref{claim:maxdegree-k} and~\ref{claim:ustar-in-W}, we have
	that the expected degree of the winner returned by mechanism {\sc Sample~and~Vote} is
	\begin{align*}
	\Exp{\delta(w,\x)} &\geq \Pr{u^*\in W}\cdot (\Delta-k) \nonumber \\ &= \left(1-\left(1-\frac{\Delta}{n-1}\right)^k\right)\left(1-\frac{1}{n}\right)^k (\Delta-k)\\
	&\geq \left(1-\left(1-\frac{\Delta}{n-1}\right)^k\right)\left(1-\frac{k}{n}\right) (\Delta-k)\nonumber \\&>\left(1-\left(1-\frac{\Delta}{n-1}\right)^k\right)\left(\Delta-2k\right)\\
	&=
	\Delta-2k-\left(1-\frac{\Delta}{n-1}\right)^k\left(\Delta-2k\right)
	\end{align*}	
	The second inequality follows by Bernoulli's inequality $(1+x)^r\geq 1+rx$ for every real $x\geq -1$ and $r\geq 0$ and the third one since $n>\Delta$. Now, the quantity $\left(1-\frac{\Delta}{n-1}\right)^k\left(\Delta-2k\right)$ is maximized for $\Delta=\frac{n-1+2k^2}{k+1}$ to a value that is at most $\frac{n+1}{k+1}-2$. Hence, 
	\begin{align*}
	\Exp{\delta(w,\x)} &\geq \Delta-2(k-1)-\frac{n+1}{k+1}.
	\end{align*}
	By setting $k\in \Theta(\sqrt{n})$, we obtain that $\Exp{\delta(w,\x)}\geq \Delta - \Theta(\sqrt{n})$, as desired.
\end{proof}


\subsection{The {\sc Sample~and~Poll} Mechanism}
\label{sec:general-upper-bound}
In the most general model, we propose the randomized mechanism {\sc Sample~and~Poll}, which is even
simpler than {\sc Sample~and~Vote}. {\sc Sample~and~Poll} forms a sample $S$ of vertices by repeating $k$ times the
selection of a vertex uniformly at random with replacement. The winner
(if any) is a vertex $w$ in $\argmax_{u\in
	{N\setminus S}}{\delta_S(u,\x)}$. We remark that, for technical reasons, we allow $S$ to be a multi-set if the same vertex is selected more than once. Then, edge multiplicities are counted in $\delta_S(u,\x)$.  Clearly, {\sc Sample~and~Poll} is impartial. The winner is decided  by the vertices in $S$, which in turn have no chance to become winners\footnote{Note that the main differences between {\sc Sample~and~Vote} and {\sc Sample~and~Poll} is (i) that the former utilizes the nominations from vertices in $N \setminus(S \cup W)$, while the latter does not, and (ii) that {\sc Sample~and~Vote} counts the nominations of vertices in the sample set $S$ only once.}. Our approximation guarantee is slightly weaker now.

\begin{theorem}
	The {\sc Sample~and~Poll} mechanism is impartial and $\Theta(n^{2/3} \ln^{1/3}n)$-additive, when $k=\left\lceil 4^{1/3} n^{2/3} \ln^{1/3}n \right\rceil$.
\end{theorem}

\begin{proof}
	Let $u^*$ be a vertex of maximum in-degree $\Delta$. If $\Delta\leq k$, {\sc Sample~and~Poll} is clearly $\Theta(n^{2/3} \ln^{1/3}n)$-additive. So, in the following, we assume that $\Delta >k$. Let $C$ be the set of vertices of in-degree at most $\Delta-k-1$. We first show that the probability $\Pr{\delta(w,\x)\leq \Delta-k-1}$ that some vertex of $C$ is returned as the winner by {\sc Sample~and~Poll} is small.
	
	Notice that if one of the vertices of $C$ is the winner, then either vertex $u^*$ belongs to to the sample set $S$ or it does not belongs to $S$ but it gets the same or fewer nominations compared to some vertex $u$ of $C$. Hence,
	
	\begin{align}\nonumber
	&\Pr{\delta(w,\x) \leq \Delta-k-1}\\\nonumber 
	&\leq \Pr{u^*\in S}+\Pr{u^*\not\in S \land \delta_S(u^*,\x)\leq \delta_S(u,\x) \mbox{ for some $u\in C$ s.t.~$u\not\in S$}}\\\nonumber
	&\leq\Pr{u^*\in S}+\sum_{u\in C}{\Pr{u^*\not\in S \land u\not\in S \land \delta_S(u^*,\x)\leq \delta_S(u,\x)}}\\\label{eq:sum-of-probs}
	&=\Pr{u^*\in S}+\sum_{u\in C}{\Pr{u^*, u\not\in S}  \Pr{\delta_S(u^*,\x)\leq \delta_S(u,\x)|u^*, u\not\in S}}
	\end{align}
	We will now bound the rightmost probability in (\ref{eq:sum-of-probs}). 
	
	\begin{clm}\label{claim:use-hoeffding}
		For every $u\in C$, $\Pr{\delta_S(u^*,\x)\leq \delta_S(u,\x)|u^*\not\in S, u\not\in S} \leq \exp\left(-\frac{k^3}{2n^2}\right).$
	\end{clm}
	
	\begin{proof}
		Assuming that $u^*$ and $u$ do not belong to the sample set $S$, we will express the difference $\delta_S(u^*,\x)-\delta_S(u,\x)$ as the sum of independent random variables $Y_i$ for $i=1, ..., k$. Variable $Y_i$ indicates the contribution of the $i$-th vertex selection to the difference $\delta_S(u^*,\x)-\delta_S(u,\x)$. In particular, $Y_i$ is equal to $1$, $-1$, and $0$ if the outgoing edges of the $i$-th vertex selected in the sample set points to vertex $u^*$ but not to vertex $u$, to vertex $u$ but not to vertex $u^*$, and either to none or to both of them, respectively. Hence, $\delta_S(u^*,\x)-\delta_S(u,\x)=\sum_{i=1}^k{Y_i}$ with $Y_i\in \{-1,0,1\}$ and 
		\begin{align*}
		\Exp{\delta_S(u^*,\x)-\delta_S(u,\x)|u^*,u\not\in S} &= \left(\Delta-\delta_u(u^*,\x)-\delta(u,\x)+\delta_{u^*}(u,\x)\right) \frac{k}{n-2} \\ &\geq \frac{k^2}{n}.
		\end{align*}
		Notice that for the computation of the expectation, we have used the facts that $\Delta-\delta_u(u^*,\x)$ vertices besides $u$ have outgoing edges pointing to $u^*$, $\delta(u,\x)-\delta_{u^*}(u,\x)$ vertices besides $u^*$ have outgoing edges pointing to $u$, and each of them is included in the sample set with probability $\frac{k}{n-2}$. The inequality follows since $\delta(u,\x)\leq \Delta-k-1$ and $\delta_u(u^*,\x), \delta_{u^*}(u,\x)\in \{0,1\}$.
		
		\noindent We will now apply Hoeffding's bound, which is stated as follows.
		
		\begin{lemma}[Hoeffding~\cite{H63}]\label{lem:hoeffding}
			Let $X_1, X_2, ..., X_t$ be independent random variables so that $\Pr{a_j\leq X_j \leq b_j} =1$. Then, the expectation of the random variable $X=\sum_{j=1}^t{X_j}$ is $\mathbb{E}[X]=\sum_{j=1}^t{\mathbb{E}[X_j]}$ and, furthermore, for every $\nu\geq 0$, $$\Pr{X \leq \mathbb{E}[X]- \nu}\leq \exp\left(-\frac{2\nu^2}{\sum_{j=1}^t{(b_j-a_j)^2}}\right).$$
		\end{lemma}
		In particular, we apply Lemma~\ref{lem:hoeffding} on the random variable $X=\delta_S(u^*,\x)-\delta_S(u,\x)$ (assuming that $u^*,u\not\in S$). Note that $t=k$, $a_j=-1$ and $b_j=1$, and recall that $\Exp{X}\geq \frac{k^2}{n}$. We obtain
		\begin{align*}
		\Pr{\delta_S(u^*,\x)-\delta_S(u,\x) \leq 0|u^*,u\not\in S} &= \Pr{X\leq 0} \\ &\leq \left(-\frac{\Exp{X}^2}{2k}\right)
		\leq \exp\left(-\frac{k^3}{2n^2}\right),
		\end{align*}
		as desired.
	\end{proof}	
	
	\noindent Using the definition of $\Exp{\delta(w,\x)}$, inequality (\ref{eq:sum-of-probs}), and Claim~\ref{claim:use-hoeffding}, we obtain
	\begin{align}\nonumber
	&\Exp{\delta(w,\x)} \geq (\Delta-k) \cdot \left(1-\Pr{\delta(w,\x)\leq \Delta-k-1}\right)\\\nonumber
	&\geq (\Delta-k)\Pr{u^*\not\in S}  \\ \nonumber &- (\Delta-k)\left(\sum_{u\in C}{\Pr{u^*, u\not\in S} \cdot \Pr{\delta_S(u^*,\x)\leq \delta_S(u,\x)|u^*, u\not\in S}}\right) \\\nonumber
	&\geq (\Delta-k) \left( 1-\frac{1}{n}\right)^k- (\Delta-k)\left( \sum_{u\in C}{\left(1-\frac{2}{n}\right)^k\cdot \exp\left(-\frac{k^3}{2n^2}\right)} \right)\\\nonumber
	&\geq (\Delta-k)\left(1-\frac{k}{n}\right) - (\Delta-k) \cdot n\cdot \exp\left(-\frac{k^3}{2n^2}\right)\\\label{eq:final-expr}
	&\geq \Delta - 2k -n^2 \cdot \exp\left(-\frac{k^3}{2n^2}\right).
	\end{align}
	The last inequality follows since $n\geq \Delta$. Setting $k=\left\lceil 4^{1/3} n^{2/3} \ln^{1/3}n \right\rceil$, (\ref{eq:final-expr}) yields $\Exp{\delta(w,\x)}\geq  \Delta-\Theta\left(n^{2/3} \ln^{1/3}n\right)$, as desired.
\end{proof}

\section{Lower Bounds}
In this section we complement our positive results by providing
impossibility results. First, in
Section~\ref{sec:lower-bound}, we provide lower bounds for a class of
mechanisms which we call strong sample mechanisms, in the single nomination model of
Holzman and Moulin~\cite{moulin13}. 
Then, in
Section~\ref{sec:general-lower-bound}, we provide a lower bound for the most general model of Alon et al.~\cite{alon11}, which applies to any deterministic mechanism.

\subsection{Strong Sample Mechanisms}\label{sec:lower-bound}

In this section, we give a characterization theorem for a class of
impartial mechanisms which we call {\em strong sample} impartial mechanisms. We
then use this characterization to provide lower bounds on the additive
approximation of deterministic and randomized mechanisms that belong
to this class. Our results suggest that the {\sc Sample and Vote} mechanism from Section~\ref{sec:random-upper} is essentially the best possible randomized mechanism in this class.

For a graph $G \in \mathcal{G}$ and a subset of vertices $S$, let $W:=W_S(G)$ be the
set of vertices outside $S$ nominated by $S$, i.e. $W=\{w\in N\setminus S:
(v,w)\in E, v\in S\}$. Then, a deterministic {\em sample} mechanism\footnote{{For simplicity we use the notation $(g,f)$ rather the more precise $(g,f(g))$. }} $(g,f)$ firstly selects
a subset $S$ using some \emph{sample function} $g: \mathcal{G} \rightarrow
2^N$, and then applies a (possibly randomized) selection mechanism $f$ by restricting its range on vertices
in $W$; notice that if $W=\emptyset$, $f$ does not select any vertex. 

This definition allows for a large class of mechanisms. For
example, the special case of sample mechanisms with $|S|=1$ (in which,
the winner has in-degree at least $1$), coincides with all negative unanimous mechanisms defined by Holzman and
Moulin~\cite{moulin13}. Indeed, when $|S|=1$, the set $W$ in never
empty and the winner has in-degree at least $1$. This is not however
the case for $|S|>1$, where $W$ could be empty when all
vertices in $S$ have outgoing edges destined for vertices in $S$ and
no winner can be declared. Characterizing all impartial sample mechanisms is an
outstanding open problem. We are able to provide a first step, by
providing a characterization for the more restricted class of impartial and strong sample mechanisms. Informally, in strong sample mechanisms,  vertices cannot
affect their chance of being selected in the sample set $S$.

\begin{definition}(Deterministic strong sample mechanisms)\label{def:strong-sample}
	We call a deterministic sample mechanism $(g,f)$ with
	sample function $g: \mathcal{G} \rightarrow 2^N$
	\emph{strong} if $g(x'_u,\x_{-u})=g(\x)$ for all $u \in g(\x)$, $x'_u \in
	N\setminus \{u\}$ and $\x \in \mathcal{G}$.
\end{definition}

The reader may observe the similarity of this definition with
impartiality (function $g$ of a strong sample mechanism satisfies
similar properties with function $f$ of an impartial selection
mechanism). The following lemma describes a straightforward, yet
useful, consequence of the above definition.

\begin{lemma}\label{LB:strongsample:plusminus}
	Let $(g,f)$ be a deterministic strong sample mechanism and let $S \subseteq
	N$. For any nomination profiles $\x,\x'$ with $\x_{-S}=\x'_{-S}$, if $S
	\setminus g(\x)\neq \emptyset$ then $S \setminus g(\x')\neq
	\emptyset$.
\end{lemma}

\begin{proof}
	For the sake of contradiction, let us assume that $S \setminus
	g(\x')= \emptyset$, i.e., the sample vertices in $\x'$ are disjoint from
	$S$. Then, by Definition~\ref{def:strong-sample}, $g(\x)$ remains
	the same as outgoing edges from vertices in $S$ should not affect the sample
	set. But then, $S \setminus g(\x)= \emptyset$, which is a contradiction. 
\end{proof}

In the next theorem, we provide a characterization for the sample
function of deterministic impartial strong sample mechanisms in the single nomination model. The theorem
essentially states that the only possible way to choose the sample set
must be independently of the graph.

\begin{theorem}\label{LB:StrongSamplek}
	In the single nomination model, any impartial deterministic strong sample mechanism $(g,f)$
	selects the sample set independently of
	the nomination profile, i.e., for all $\x,\x' \in \mathcal{G}^1$,
	$g(\x)=g(\x')=S$.
\end{theorem}

\begin{proof}
	Since we are in the single nomination model, without loss of generality, we can use the simplified notation $x_{u}=v$ instead of $x_{u}=\{(u,v)\}$ for any graph $\x \in \mathcal{G}^1$.
	Consider any sample mechanism $(g,f)$ and any nomination profile $\x \in \mathcal{G}^1$. It suffices to show 
	that for any vertex $u$, and any alternative vote $x'_u$, the
	sample set must remain the same, i.e., $g(x'_u,\x_{-u},)=g(\x)$. If
	$u\in g(\x)$, this immediately follows by
	Definition~\ref{def:strong-sample}.  In the following, we prove 
	two claims showing that this holds also when $u \notin g(\x)$;
	Claim~\ref{LB:StrongSamplek:lemma1} treats the case where $u$ is a
	winner of a profile, while Claim~\ref{StrongSamplek:lemma2} treats
	the case where $u$ is a not a winner.
	
	\begin{clm}\label{LB:StrongSamplek:lemma1}
		Let $(g,f)$ be an impartial deterministic strong sample mechanism
		and let $\x$ be any nomination profile in $\mathcal{G}^1$. Then the sample set must
		remain the same for any other vote of the winner, i.e.,
		$g(\x)=g(x'_{f(\x)},\x_{-f(\x)})$ for any $x'_{f(\x)} \in N
		\setminus \{f(\x)\}$.
	\end{clm}

	\begin{proof}
		
		Let $w=f(\x)$ be the winner, for some nomination profile $\x$. We will prove the claim by induction on the in-degree of the winner, $\delta(w,\x)$. Note that $\delta(w,\x)>0$ for any sample mechanism and any $\x \in \mathcal{G'}$.
		
		{\bf(Base case: $\delta(w,\x)=1$)}  Let $S=g(\x)$ be the sample set for profile $\x$.
		Assume for the sake of contradiction that when $w$ changes its
		vote to $x'_w$, the sample for profile $\x'=(x'_w,\x_{-w})$ changes, i.e.,
		$g(\x')=S'\neq S$. We first note that impartiality of $f$ implies
		that $w=f(\x')$. Next, observe that the vertex voting for $w$ in $S$
		must be also in $S'$; otherwise, $w$ becomes a winner without
		getting any vote from the sample set, which contradicts our
		definition of sample mechanisms. We will show that this must be the case
		for all vertices in $S$.
		
		To do this, we will expand two parallel branches, creating a sequence of
		nomination profiles starting from $\x$ and $\x'$ which will
		eventually lead to a contradiction. 
		Figure \ref{fig:LB:X} depicts the situation for $\x$ and $\x'$.

		\begin{figure}[h]
	\subfloat[profile $\x$]{\label{fig:LB:x}
		\centering
		\begin{tikzpicture}[winner/.style={circle,draw=black!80,fill=black!80!,very thick,minimum size=0.8cm,text=white},simple/.style={very thick,circle,draw=black!80,minimum size=0.8cm},sample/.style={circle, very thick,draw=black!80!red,fill=black!30!,minimum size=0.8cm, dashed},cand/.style={circle,draw=red!80,fill=black!10,minimum size=0.8cm,very thick},,scale=1, every node/.style={transform shape}]
		
		\node[simple] at  (-1.5,1) (a) {$a$};
		\node[simple] at  (1.5,1) (b) {$b$};
		\node[winner] at  (0,0) (w) {$w$};
		
		\node[sample] at  (-3,-2) (s) {$s$};
		\node[sample] at  (-1.5,-2) (s1) {$ $};
		\node[sample] at  (0,-2) (s2) {$ $};
		
		\node[simple] at  (1.5,-2) (sp1) {$ s' $};
		\node[simple] at  (3,-2) (sp2) {$ $};
		
		\draw[-Latex ] (s) to (w);		
		\draw[-Latex ] (w) to (a);
		\end{tikzpicture}
	}\hfill
	\subfloat[profile $\x'$]{\label{fig:LB:xp}
		\centering
		\begin{tikzpicture}[winner/.style={circle,draw=black!80,fill=black!80!,very thick,minimum size=0.8cm,text=white},simple/.style={very thick,circle,draw=black!80,minimum size=0.8cm},sample/.style={circle, very thick,draw=black!80!red,fill=black!30!,minimum size=0.8cm, dashed},cand/.style={circle,draw=red!80,fill=black!10,minimum size=0.8cm,very thick},,scale=1, every node/.style={transform shape}]
		\node[simple] at  (-1.5,1) (a) {$a$};
		\node[simple] at  (1.5,1) (b) {$b$};
		\node[winner] at  (0,0) (w) {$w$};
		
		\node[sample] at  (-3,-2) (s) {$s$};
		\node[simple] at  (-1.5,-2) (s1) {$ $};
		\node[simple] at  (0,-2) (s2) {$ $};
		
		\node[sample] at  (1.5,-2) (sp1) {$ s' $};
		\node[sample] at  (3,-2) (sp2) {$ $};
		
		\draw[-Latex ] (s) to (w);		
		\draw[-Latex ] (w) to (b);
		\end{tikzpicture}
	}
	\caption{The starting profiles $\x$ and $\x'$ in Claim~\ref{LB:StrongSamplek:lemma1}. The dark vertex is the winner, while the light, dashed-lined vertices are the members of the sets $S$ and $S'$, respectively.}
	\label{fig:LB:X}
\end{figure}

		We start with the profile $\x'$. Consider a vertex $s' \in S' \setminus S$. We create
		a profile $\z'$ in which all vertices in $S'\setminus s'$ vote for $s'$
		(i.e., $z_v=s'$, for each $v\in S'\setminus s'$), vertex $s'$ votes for
		$w$ (i.e., $z_v=w$), while the rest of the vertices vote as in $\x'$
		(i.e., $z_v=x_v$, for each $v\not \in S'$). For illustration, see
		Figures \ref{fig:LB:z} and \ref{fig:LB:zp}. By the definition of a
		strong sample mechanism, we obtain $g(\z')=g(\x')$, since only votes
		of vertices in $S'$ have changed. Notice also that $f(\z')=w$, as this is
		the only vertex outside $S'$ that receives votes from $S'$. We now move
		to profile $\x$ and apply the same sequence of deviations, involving all the
		vertices in $S'$. These lead to the profile $\z$, which differs from
		$\z'$ only in the outgoing edge of vertex $w$.
		
		\begin{figure}[tp]
	\subfloat[profile $\z$]{\label{fig:LB:z}   
		\centering
		\begin{tikzpicture}[winner/.style={circle,draw=black!80,fill=black!80!,very thick,minimum size=0.8cm},simple/.style={very thick,circle,draw=black!80,minimum size=0.8cm},sample/.style={circle, very thick,draw=black!80!red,fill=black!40!green!60,minimum size=0.8cm, dashed},cand/.style={circle,draw=red!80,fill=black!10,minimum size=0.8cm,very thick},scale=0.7, every node/.style={transform shape}]
		
		\node[simple] at  (-1.5,1) (a) {$a$};
		\node[simple] at  (1.5,1) (b) {$b$};
		\node[simple,diamond,draw=black,minimum size=1cm] at  (0,0) (w) {$w$};
		
		\node[simple] at  (-3,-2) (s) {$s$};
		\node[simple] at  (-1.5,-2) (s1) {$ $};
		\node[simple] at  (0,-2) (s2) {$ $};
		
		\node[simple,dashed,draw=black,diamond,minimum size=1cm] at  (1.5,-2) (sp1) {$ s' $};
		\node[simple] at  (3,-2) (sp2) {$ $};
		
		\draw[-Latex, bend right ] (s) to (sp1);
		\draw[-Latex ] (sp1) to (w);
		\draw[-Latex ] (sp2) to (sp1);
		
		\draw[-Latex ] (w) to (a);
		\end{tikzpicture}
	}
	\hfill
	\subfloat[profile $\z'$]{\label{fig:LB:zp}
		\centering
		\begin{tikzpicture}[winner/.style={circle,draw=black!80,fill=black!80!,very thick,minimum size=0.8cm,text=white},simple/.style={very thick,circle,draw=black!80,minimum size=0.8cm},sample/.style={circle, very thick,draw=black!80!red,fill=black!30!,minimum size=0.8cm, dashed},cand/.style={circle,draw=red!80,fill=black!10,minimum size=0.8cm,very thick},scale=0.7, every node/.style={transform shape}]
		
		\node[simple] at  (-1.5,1) (a) {$a$};
		\node[simple] at  (1.5,1) (b) {$b$};
		\node[winner] at  (0,0) (w) {$w$};
		
		\node[sample] at  (-3,-2) (s) {$s$};
		\node[simple] at  (-1.5,-2) (s1) {$ $};
		\node[simple] at  (0,-2) (s2) {$ $};
		
		\node[sample] at  (1.5,-2) (sp1) {$ s' $};
		\node[sample] at  (3,-2) (sp2) {$ $};
		
		\draw[-Latex, bend right ] (s) to (sp1);
		\draw[-Latex ] (sp1) to (w);
		\draw[-Latex ] (sp2) to (sp1);
		
		\draw[-Latex ] (w) to (b);
		\end{tikzpicture}
	}\\
	\hfill
	\subfloat[profile $\y$]{\label{fig:LB:y}
		\begin{tikzpicture}[winner/.style={circle,draw=black!80,fill=red!80!,very thick,minimum size=0.8cm},simple/.style={very thick,circle,draw=black!80,minimum size=0.8cm},sample/.style={circle, very thick,draw=black!80!red,fill=black!30!,minimum size=0.8cm, dashed},cand/.style={circle,draw=red!80,fill=black!10,minimum size=0.8cm,very thick},scale=0.7, every node/.style={transform shape}]
		
		\node[simple] at  (-1.5,1) (a) {$a$};
		\node[simple] at  (1.5,1) (b) {$b$};
		\node[simple,draw=black,diamond,minimum size=1cm] at  (0,0) (w) {$w$};
		
		\node[simple] at  (-3,-2) (s) {$s$};
		\node[simple] at  (-1.5,-2) (s1) {$ $};
		\node[simple] at  (0,-2) (s2) {$ $};
		
		\node[sample] at  (1.5,-2) (sp1) {$ s' $};
		\node[simple,dashed,draw=black,diamond,minimum size=1cm] at  (3,-2) (sp2) {$v$};
		
		\draw[-Latex, bend right ] (s) to (sp1);
		\draw[-Latex ] (sp1) to (sp2);
		\draw[-Latex ] (sp2) to (w);
		
		\draw[-Latex ] (w) to (a);
		\end{tikzpicture}   	
	}
	\hfill
	\subfloat[profile $\y'$]{\label{fig:LB:yp}   
		\begin{tikzpicture}[winner/.style={circle,draw=black!80,fill=black!80!,very thick,minimum size=0.8cm,text=white},simple/.style={very thick,circle,draw=black!80,minimum size=0.8cm},sample/.style={circle, very thick,draw=black!80!red,fill=black!30!,minimum size=0.8cm, dashed},cand/.style={circle,draw=red!80,fill=black!10,minimum size=0.8cm,very thick},scale=0.7, every node/.style={transform shape}]
		
		\node[simple] at  (-1.5,1) (a) {$a$};
		\node[simple] at  (1.5,1) (b) {$b$};
		\node[winner] at  (0,0) (w) {$w$};
		
		\node[sample] at  (-3,-2) (s) {$s$};
		\node[simple] at  (-1.5,-2) (s1) {$ $};
		\node[simple] at  (0,-2) (s2) {$ $};
		
		\node[sample] at  (1.5,-2) (sp1) {$ s' $};
		\node[sample] at  (3,-2) (sp2) {$v$};
		
		\draw[-Latex, bend right ] (s) to (sp1);
		\draw[-Latex ] (sp1) to (sp2);
		\draw[-Latex ] (sp2) to (w);
		
		\draw[-Latex ] (w) to (b);
		\end{tikzpicture}
	}
	\caption{Profiles $\z$ and $\z'$ in the base case of the proof of Claim~\ref{LB:StrongSamplek:lemma1}: if $v=s'$ then $s' \notin g(\z)$ and since this is the only vertex voting for $w$, $w$ cannot win in $\z$, while it must be the winner in $\z'$ ---a contradiction. Profiles $\y$ and $\y'$: if $ s' \in g(\y)$, we let $s'$ vote for $v$ and $v$ for $w$, making $w$ the winner in $\y'$ but not in $\y$.
		A dark circle denotes the winner, while light, dashed-lined circles denote the members of the sample sets $S$ and $S'$.
		A solid-lined diamond denotes a vertex that cannot be the winner and a dashed-lined diamond denotes a vertex that cannot be in the sample set.}
	\label{fig:LB:Z}
\end{figure}

		By Lemma~\ref{LB:strongsample:plusminus}, there is a vertex $v \in
		S'$ such that $v \notin g(\z)$. If $v=s'$, then we end up in a
		contradiction. This is because $f(\z) \neq w$, since $s'$ is the only
		vertex voting for $w$ in $\z'$ and $s'$ is not in the sample, while $f(\z') =  w$,
		as stated by the other branch and since, when $w$ change its vote to $x'_w$, the created profile is
		$(x'_w,\z_{-w})=\z'$ contradicting impartiality (see also
		Figures \ref{fig:LB:z} and \ref{fig:LB:zp}).

		We are now left with the case where $s'\in g(\z)$ and $v\neq
		s'$. Starting from $\z$ and $\z'$, we will create profiles $\y$ and
		$\y'$ (see Figures \ref{fig:LB:y} and \ref{fig:LB:yp}) as follows: we construct $\y$ by letting $s'$ 
		vote towards $v$ (i.e., $y_{s'}=v$), $v$ vote towards $w$ (i.e., $y_v=w$) and
		$y_i=z_i$ for all other vertices $i\neq v,s'$. By the strong sample
		property, when $s'$ votes towards $v$ the sample set is preserved,
		i.e., $v$ cannot get in the sample. Also, when $v$ votes, $v$ cannot
		get in the sample (by a trivial application of Lemma
		\ref{LB:strongsample:plusminus}); therefore, $v\not \in g(\y)$. Hence, 
		$w$ cannot be the winner as its only incoming vote is from $v$, a vertex that does not
		belong to the sample set $g(\y)$.
		
		Starting from $\z'$, we create similarly $\y'$ by letting $s'$ vote
		towards $v$ ($y'_{s'}=v$), $v$ to vote towards $w$ ($y'_v=w$) and
		$y'_i=z_i$ for all other vertices $i\neq v,s'$. In this case, $S'$ will be
		preserved as sample set in profile $\y'$ (i.e. $g(\y')=S'$). Therefore, $w$ is the only vertex voted by the sample set
		and must be the winner, leading to a contradiction (see Figures
		\ref{fig:LB:y} and \ref{fig:LB:yp}).
		
		{\bf (Induction step)} Assume as induction hypothesis that, for all
		profiles $\x \in \mathcal{G}^1$, it holds $g(\x)=g(x'_w,\x_{-w})=S$ when
		$\delta(w,\x)\leq \lambda$, for some $\lambda \geq 1$. Now, consider
		any profile $\x$ where $f(\x)=w$ and $\delta(w,\x)= \lambda+1$ and
		assume for the sake of contradiction that there is some graph
		$\x'=(x'_w,\x_{-w})$ where $g(\x')=S' \neq S$. Without loss of generality, let
		$\delta_{S}(w,\x) \leq \delta_{S'}(w,\x')$.

		Starting from $\x'$, we create profile $\z'$, by letting all vertices in
		$S'$ vote for some $s' \in S'$ and $s'$ vote for $w$, i.e., $z'_v
		= s'$ for each vertex $v \in S'\setminus\{s'\}$ and $z'_{s'}=w$. The
		strong sample property implies that $g(\z')=S'$ and
		$f(\z')=w$. We focus now on profile $\x$, and create the profile $\z$, by performing the same series of deviations, i.e., by letting all vertices in $S'\setminus{s'}$ vote for $s'$ and $s'$ vote for
		$w$. Note here that $\z$ differs from $\z'$ only in the outgoing edge of $w$.
		Like before, Lemma~\ref{LB:strongsample:plusminus} establishes
		that there will be some vertex $v \in S'$ such that $v \notin g(\z)$, i.e., $g(\z) \neq S'$.
		Turning our attention back to $\z'$, we let $w$ change its vote to
		$x_w$, creating profile $(x_w,\z'_{-w})$. Observe that $(x_w,\z'_{-w})=\z$. When $\delta(w,\z') <
		\delta(w,\x)$, by the induction hypothesis we have $g(\z)=S'$, a
		contradiction. 
		
		We need also to  handle the case $\delta(w,\z') = \delta(w,\x)$. We will use a series of careful steps to decrease the in-degree of $w$, without changing the sample set. This will allow us to use the induction hypothesis to finalize our proof. 
		
		Let $L$ denote the set of vertices which vote for $w$ in profile $\x$. We note here that, we may end up in the case $\delta(w,\z') = \delta(w,\x)$ only because a single vertex votes for $w$ in $\z'$, i.e. $|g(\z') \cap L|=1$; otherwise we could decrease the in-degree of $w$ in $\z'$ without changing the sample set and directly use the induction hypothesis to prove the claim. The aforementioned vertex $s'$ is the single vertex in $g(\z') \cap L$. Note here that there exists at least one vertex in $g(\x) \cap L$. Say this is vertex $s$. If $\delta(s,\z') \leq \lambda-1$, we can create the profile $\y'$, where $s'$ votes for $s$ (i.e. $y'_{s'}=s$ and $\y'=(y'_{s'},\z'_{-s'})$), hence $f(\y')=s$ and $g(\y')=g(\z')=S'$ (recall that $s'$ is the single vertex in $g(\z')$ voting outside of $g(\z')$ ). We can create now the profile $\q'$ where vertex $s$ votes for vertex $s'$ (i.e. $q'_s=s'$) and the other vertices vote like in $\y'$, i.e. $\q'=(q'_s,\y'_{-s})$. Since $\delta(s,\y') \leq \lambda$ and $f(\y')=s$, by changing the outgoing edge of the winning vertex $s$, the sample set does not change, due to the induction hypothesis, i.e. $g(\q')=g(\y')=S'$. Finally, we create the profile $\rr'$, where $s'$ votes for $w$ (i.e. $r'_{s'}=w$) and the other vertices vote like in $\q'$ (i.e. $\rr'=(r'_{s'},\q'_{-s'})$). The strong sample property now implies that $g(\rr')=g(\q')=S'$ and $f(\rr')=w$. Since $\delta(w,\rr')= \lambda$, we can invoke the induction hypothesis once again: we can create profile $\rr$ by letting $w$ vote as in $\x$ (i.e. $\rr=(x_w,\rr'_{-w})$) and  $g(\rr)=S'$. 
		
		At this point, we reverse our previous moves. First, we create the profile $\q$ by allowing $s'$ change its vote for $s$ (i.e. $q_{s'}=s$ and $\q=(q_{s'},\rr_{-s'})$). Again, the strong sample property implies that $g(\q)=g(\rr)=S'$, which results to $f(\q)=s$. Finally, we create profile $\y$ by letting $s$ vote for $w$ (i.e. $y_{s}=w$ and $\y=(y_s,\q_{-s})$). The induction hypothesis implies now that $g(\y)=S'$. Observe that $\y$ is indeed profile $\z$ (i.e. $\y=\z$), which is identical to $\z'$, except from the vote of $w$, hence $g(\z)=S'$. Recall however, that $g(\z)\neq S'$, a contradiction. Given that $\delta(s,\z') \leq \lambda-1$, the claim follows.
				
				\begin{figure}[htp]
			\centering
			\subfloat[profile $\rr^{(0)'}$]{\label{fig:LB:OneSamplerIn:1}
				\centering
				\begin{tikzpicture}[winner/.style={circle,draw=black!80,fill=black!80!,very thick,minimum size=0.8cm,text=white},simple/.style={very thick,circle,draw=black!80,minimum size=0.8cm},sample/.style={circle, very thick,draw=black!80!red,fill=black!30!,minimum size=0.8cm, dashed}, cand/.style={circle,draw=red!80,fill=black!10,minimum size=0.8cm,very thick},node distance=2cm, minimum size=1cm,scale=0.75]
				\node[winner] at  (0,0) (w) {$w$};
				
				\node[simple] (s1) [below left of=w]  {$ $};
				\node[simple] (s) [below of=w]{$ s$};
				\node[sample] (sp) [below right of=w]{$ s' $};
				
				\node (sdots) [below right of =sp] {$\cdots$};
				\node (sc2) [below of =s1] {$\cdots$};
				\node[simple]  (sc3) [below of= s ] {$s_1 $};
				\node[simple]  (sc4) [below right of= s ]{$s_2 $};
				
				\node at  (-2.5,-6) (sc21) {$ $};
				\node[simple]  (sc31) [below of=sc3] {$ $};
				\node[simple]  (sc41) [below of=sc4] {$s_3$};

				\draw[-Latex,dashed ] (s1) to (w);
				\draw[-Latex,dashed ] (sp) to (w);
				\draw[-Latex, thick ] (s) to (w);
				\draw[-Latex,dashed ] (sc2) to (s);
				\draw[-Latex, thick] (sc3) to (s);
				\draw[-Latex, thick] (sc4) to (s);
				\draw[-Latex,dashed ] (sc31) to (sc3);
				\draw[-Latex, thick ] (sc41) to (sc4);
				\draw[-Latex,dashed ] (w) to (sc4);
				\draw[-Latex,dashed ] (sdots) to (sp);
				\end{tikzpicture}
			}
			\hfill
			\subfloat[profile $\y^{(1)'}$]{\label{fig:LB:OneSamplerIn:2}
					\centering
				\begin{tikzpicture}[winner/.style={circle,draw=black!80,fill=black!80!,very thick,minimum size=0.8cm,text=white},simple/.style={very thick,circle,draw=black!80,minimum size=0.8cm},sample/.style={circle, very thick,draw=black!80!red,fill=black!30!,minimum size=0.8cm, dashed}, cand/.style={circle,draw=red!80,fill=black!10,minimum size=0.8cm,very thick},node distance=2cm, minimum size=1cm,scale=0.75]
				
				\node[simple] at  (0,0) (w) {$w$};
				
				\node[simple] (s1) [below left of=w]  {$ $};
				\node[simple] (s) [below of=w]{$ s$};
				\node[sample] (sp) [below right of=w]{$ s' $};
				
				\node (sdots) [below right of =sp] {$\cdots$};
				\node (sc2) [below of =s1] {$\cdots$};
				\node[simple]  (sc3) [below of= s ] {$s_1 $};
				\node[simple]  (sc4) [below right of= s ]{$s_2 $};
				
				\node at  (-2.5,-6) (sc21) {$ $};
				\node[simple]  (sc31) [below of=sc3] {$ $};
				\node[winner]  (sc41) [below of=sc4] {$s_3$};

				\draw[-Latex,dashed ] (s1) to (w);
				\draw[-Latex,bend left,dashed ] (sp) to (sc41);
				\draw[-Latex, thick ] (s) to (w);
				\draw[-Latex,dashed ] (sc2) to (s);
				\draw[-Latex, thick] (sc3) to (s);
				\draw[-Latex, thick] (sc4) to (s);
				\draw[-Latex, dashed ] (sc31) to (sc3);
				\draw[-Latex, thick ] (sc41) to (sc4);
				\draw[-Latex,dashed ] (w) to (sc4);
				\draw[-Latex,dashed ] (sdots) to (sp);
				\end{tikzpicture}
			}\\
			\subfloat[profile $ \q^{(1)'}$]{\label{fig:LB:OneSamplerIn:3}
				\begin{tikzpicture}[winner/.style={circle,draw=black!80,fill=black!80!,very thick,minimum size=0.8cm,text=white},simple/.style={very thick,circle,draw=black!80,minimum size=0.8cm},sample/.style={circle, very thick,draw=black!80!red,fill=black!30!,minimum size=0.8cm, dashed}, cand/.style={circle,draw=red!80,fill=black!10,minimum size=0.8cm,very thick},node distance=2cm, minimum size=1cm,scale=0.75]
				
				\node[simple] at  (0,0) (w) {$w$};
				
				\node[simple] (s1) [below left of=w]  {$ $};
				\node[simple] (s) [below of=w]{$ s$};
				\node[sample] (sp) [below right of=w]{$ s' $};
				
				\node (sdots) [below right of =sp] {$\cdots$};
				\node (sc2) [below of =s1] {$\cdots$};
				\node[simple]  (sc3) [below of= s ] {$s_1 $};
				\node[simple]  (sc4) [below right of= s ]{$s_2 $};
				
				\node at  (-2.5,-6) (sc21) {$ $};
				\node[simple]  (sc31) [below of=sc3] {$ $};
				\node[winner]  (sc41) [below of=sc4] {$s_3$};

				\draw[-Latex ,dashed] (s1) to (w);
				\draw[-Latex,bend left,dashed ] (sp) to (sc41);
				\draw[-Latex, thick ] (s) to (w);
				\draw[-Latex,dashed ] (sc2) to (s);
				\draw[-Latex, thick] (sc3) to (s);
				\draw[-Latex, thick] (sc4) to (s);
				\draw[-Latex,dashed ] (sc31) to (sc3);
				\draw[-Latex, bend right,dashed ] (sc41) to (sp);
				\draw[-Latex,dashed ] (w) to (sc4);
				\draw[-Latex,dashed ] (sdots) to (sp);
				\end{tikzpicture}   
			}
			\hfill
			\subfloat[profile $\rr^{(4)'}$]{\label{fig:LB:OneSamplerIn:4}
				\begin{tikzpicture}[winner/.style={circle,draw=black!80,fill=black!80!,very thick,minimum size=0.8cm,text=white},simple/.style={very thick,circle,draw=black!80,minimum size=0.8cm},sample/.style={circle, very thick,draw=black!80!red,fill=black!30!,minimum size=0.8cm, dashed}, cand/.style={circle,draw=red!80,fill=black!10,minimum size=0.8cm,very thick},node distance=2cm, minimum size=1cm,scale=0.75]
				
				\node[winner] at  (0,0) (w) {$w$};
				
				\node[simple] (s1) [below left of=w]  {$ $};
				\node[simple] (s) [below of=w]{$ s$};
				\node[sample] (sp) [below right of=w]{$ s' $};
				
				\node (sdots) [below right of =sp] {$\cdots$};
				\node (sc2) [below of =s1] {$\cdots$};
				\node[simple]  (sc3) [below of= s ] {$s_1 $};
				\node[simple]  (sc4) [below right of= s ]{$s_2 $};
				
				\node at  (-2.5,-6) (sc21) {$ $};
				\node[simple]  (sc31) [below of=sc3] {$ $};
				\node[simple]  (sc41) [below of=sc4] {$s_3$};

				\draw[-Latex,dashed ] (s1) to (w);
				\draw[-Latex,dashed ] (sp) to (w);
				\draw[-Latex, dashed ] (s) to (sp);
				\draw[-Latex,dashed ] (sc2) to (s);
				\draw[-Latex,dashed  ] (sc3) to (sp);
				\draw[-Latex,dashed] (sc4) to (sp);
				\draw[-Latex ,dashed] (sc31) to (sc3);
				\draw[-Latex, bend right,,dashed ] (sc41) to (sp);
				\draw[-Latex,dashed ] (w) to (sc4);
				\draw[-Latex,dashed ] (sdots) to (sp);
				\end{tikzpicture}
			}
			\caption{
				Induction step for the proof of Claim~\ref{LB:StrongSamplek:lemma1}. An example with $\lambda=2$. To use the induction hypothesis, we need to decrease the in-degree of $w$ to $2$. In figure~\ref{fig:LB:OneSamplerIn:1}, vertex $s'$ is the single vertex in the sample, voting for the winner $w$. The other vertices in the sample vote for $s'$. To decrease the in-degree of $w$, we identify the tree $T$, denoted by the solid thick edges. First, we let $s'$ vote for $s_3$, which is now the winner (figure \ref{fig:LB:OneSamplerIn:2}). Then we let $s_3$ vote for $s'$. By impartiality $s_3$ (with in-degree of $1$) retains its winner status and the sample set is still the same (figure \ref{fig:LB:OneSamplerIn:3}).  The same procedure continues until all edges of $T$ are redirected to $s'$ and the in-degree of $w$ decreases to $2$. During this process, the sample set remains invariant. (figure $\ref{fig:LB:OneSamplerIn:4}$). The dark vertex denotes the winner, while the light, dashed-lined vertices denote members of the sample set.}
			\label{fig:LB:OneSamplerInL}
		\end{figure}

		If $\delta(s,\z') > \lambda-1$, we generalize the idea described in the previous two paragraphs. We will revert enough edges towards $s'$, for the in-degree of $s$ to become equal to $\lambda-1$. To do this we identify a tree $T$ in $\z'$, with vertex $s$ as the root: We start from the vertices voting towards $s$ and we select $\delta(s,\z)-(\lambda-1)$ of them as children of $s$: first, we select vertices with in-degree at most $\lambda-1$. If we end up with vertices with higher in-degree than $\lambda-1$, we repeat the process for each child, until all leafs in the tree have in-degrees at most $\lambda-1$. This is assured, since we are in the single nomination model and each vertex belongs in at most one directed cycle. 
		
		Let $k$ be the number of vertices in the tree $T$ and let $\rr^{(0)'}=\z'$; hence $g(\rr^{(0)'})=S'$. Starting from an arbitrary leaf on $T$, let $v_i$ denote the $i$-th vertex we visit. For each $i \in \{1,...,k\}$, we create three profiles: $\y^{(i)'},\q^{(i)'}$ and $\rr^{(i)'}$. First, we create profile $\y^{(i)'}$ by letting $s'$ vote for $v_i$ (i.e. $y^{(i)'}_{s'}=v_i$) and the remaining vertices vote like in $\rr^{(i-1)'}$, i.e. $\y^{(i)'}=(y^{(i)'}_{s'},\rr^{(i-1)'}_{-v_i})$. Due to the strong sample property, $g(\y^{(i)'})=g(\rr^{(i-1)'})$. Also, $f(\y^{(i)'})=v_i$ since $v_i$ is the only vertex voted by the sample. Then we create the profile $\q^{(i)'}$, where $q^{(i)'}_{v_i}=s'$ and the other vertices vote like $\y^{(i)'}$, i.e. $\q^{(i)'}=(q^{(i)'}_{v_i},\y^{(i)'}_{-i})$. Due to impartiality $f(\q^{(i)'})=v_i$. If we are traversing the tree $T$ from the leaves to the root, each vertex $v_i$ has in-degree at most $\lambda$ and by the induction hypothesis $g(\q^{(i)'})=g(\y^{(i)'})=S'$. Finally, we create the profile $\rr^{(i)'}$ by letting $s'$ vote for $w$, i.e $r^{(i)'}_{s'}=w$ and $\rr^{(i)'}=(r^{(i)'}_{s'},\y^{(i)'}_{-s'})$. We traverse the vertices starting from a leaf, and after visiting all vertices in the same level, we pass to the next level and we keep the order of the vertices visited. Note that in all these changes the \emph{sample set does not change} and each vertex in $T$ (including vertex $s$, which is traversed last, i.e. $v_k=s$) has in-degree at most $\lambda$. An example of this process is depicted in Figure~\ref{fig:LB:OneSamplerInL}.
		
		At this point, we start a reverse procedure. We first create the profile $\rr^{(k)}=(x_w,\rr^{(k)'}_{-w})$, where we let vertex $w$ to vote like in profile $\x$. By the induction hypothesis, $g(\rr^{(k)})=g(\rr^{(k)'})=S'$, since $f(\rr^{(k)'})=w$ and $\delta(w,\rr^{(k)'}) \leq \lambda$. We then start to traverse the vertices in tree $T$ on the opposite direction, i.e. $v_k,v_{k-1},...,v_{1}$. For each $i \in \{1,...,k\}$ We create a similar series of profiles, where the sample set will remain invariant. Starting from  $\rr^{(i)}$ we create the profile $\q^{(i)}$, where $s'$ votes towards $q^{(i)'}_{s'}$, i.e. $\q^{(i)}=(q^{(i)'}_{s'},\rr^{(i)}_{-s'})$. Due to the strong sample property  $g(\q^{(i)})=g(\rr^{(i)})=S'$ and $f(\q^{(i)})=q^{(i)'}_{s'}$. Observe that $q^{(i)}$ and $q^{(i)'}$ differ only in the outgoing edge of $w$. As a result $\delta(q^{(i)'}_{s'},\q^{(i)}) \leq \lambda$. We create now the profile $\y^{(i)}$ where  $q^{(i)'}_{s'}$, the winning node in $\q^{(i)}$, votes towards $y^{(i)'}_{q^{(i)'}_{s'}}$, i.e. $y^{(i)}_{q^{(i)'}_{s'}}=y^{(i)'}_{q^{(i)'}_{s'}}$  and $\y^{(i)} = ( y^{(i)}_{q^{(i)'}_{s'}},\q^{(i)}_{-q^{(i)'}_{s'}}) $. Again $\y^{(i)}$ and $\y^{(i)'}$ defer only in the vote of $w$. Because of the induction hypothesis $g(\y^{(i)})=g(\q^{(i)})=S'$.	
		Finally, we revert $s'$ towards $w$ and create the profile $\rr^{(i-1)}$ such that $r^{(i-1)}_{s'}=w$ and $\rr^{(i-1)}=(r^{(i-1)}_{s'}, \y^{(i)}_{-s'})$. Again $g(\rr^{(i-1)})=g(\y^{(i)})=S'$. 
		
		After this series of changes, we end up in profile $\rr^{(0)}$, which differs from $\rr^{(0)'}$ only in the outgoing edge of $w$. Since in all changes described above the sample set remains invariant, then $g(\rr^{(0)})=S'$.
		Observe now that $\rr^{(0)}=\z$, for which we know that $g(\z) \neq S'$, a contradiction. This concludes the proof.
	\end{proof}
	
	\noindent The next claim establishes the remaining case, that no vertex $u\not
	\in g(\x), u\neq f(\x)$ can change the sample set.
	
	\begin{clm}\label{StrongSamplek:lemma2}
		Let $(g,f)$ be an impartial deterministic strong sample mechanism,
		$\x$ be a nomination profile in $\mathcal{G}^1$ and $u$ a vertex with $u\not \in g(\x), u\neq f(\x)$. Then
		$g(\x)=g(x'_{u},\x_{-u})$ for any other vote $x'_{u} \in N
		\setminus \{{u}\}$.
	\end{clm}

	\begin{proof}
		For the sake of contradiction, consider any profile $\x \in \mathcal{G}^1$ and assume that there exists some
		nomination profile $\x'=(x'_{u},\x_{-u})$ with $g(\x')=S' \neq  g(\x)$.
		Starting from $\x'$, we define a profile $\z'$ in which all vertices in
		$S'$ vote for $u$, and the rest vote as in $\x'$. That is, $z'_v=u$,
		for all $v\in S'$ and $z'_v=x'_v$ otherwise. Clearly $f(\z')=u$, as
		all the sample vertices vote for $u$. By
		Claim~\ref{LB:StrongSamplek:lemma1}, we know that
		$g(x_u,\z'_{-u})=g(\z')=S'$.
		
		Starting from $\x$, we define a profile $\z$ in which all vertices in
		$S'$ vote for $u$, and the rest vote as in $\x$. Since $S'\neq
		g(\x)$, by Lemma~\ref{LB:strongsample:plusminus}, we get $g(\z)
		\neq S'$. Observe that $\z=(x_u,\z'_{-u})$, which leads to a
		contradiction.
	\end{proof}
	
	\noindent This completes the proof of Theorem~\ref{LB:StrongSamplek}.
\end{proof}

We next use Theorem~\ref{LB:StrongSamplek} to obtain lower bounds on
the additive approximation guarantee obtained by any deterministic strong sample mechanisms.

\begin{corollary}
	There is no impartial deterministic strong sample mechanism with additive approximation better than $n-2$ in the single nomination model.
\end{corollary}

\begin{proof}
	Let $S$ be the sample set which, by Theorem~\ref{LB:StrongSamplek},
	must be selected independently of $\x$, and let $v\in S$. Define $\x$
	so that all vertices in $N\setminus \{v\}$ vote for $v$ and all other vertices have in-degree either $0$ or $1$. Then, $\Delta(\x)=n-1$,
	but the mechanism selects a vertex of in-degree exactly $1$.  
\end{proof}

We remark that the strong sample mechanism that uses a specific vertex as singleton sample achieves this additive approximation guarantee.

Our next step, is to extend the notion of sample mechanisms to
randomized variants and provide a lower bound on their additive
approximation guarantee, which shows that {\sc Sample and Vote} (with
$k=\Theta(\sqrt{n})$; see Section~\ref{sec:random-upper}) is an
optimal mechanism from this class. We next define the family of
randomized strong sample mechanisms.
\begin{definition}(Randomized strong sample mechanisms)\label{StrongSample:RandDef}
	A randomized strong sample mechanism $(g,f)$ is a probability distribution over a family $\{(g_i,f_i): i \in \mathbb{N} \}$ of strong sample mechanisms.
\end{definition}

Note that {\sc Sample and Vote} and {\sc Sample and Poll} are both randomized strong sample mechanisms: For a given $k$, each of the possible sample sets define a deterministic sample mechanism, and the winner (if any) belongs in the set $W$. This is however not the case for more complex mechanisms like those appearing in \cite{bousquet2014} and in \cite{fischer2015}.

\begin{corollary}\label{StrongSample:Rand}
	There is no impartial randomized strong sample mechanism with additive approximation better than
	$\Omega(\sqrt{n})$ in the single nomination model.
\end{corollary}

\begin{proof}
	By Theorem~\ref{LB:StrongSamplek}, in any deterministic strong sample mechanism, the sample set is the same for any input graph $\x$. Hence, in a randomized strong sample mechanism, the probability that a vertex $u$ belongs in the sample set, is affected only by the sample functions used by the mechanism. As such, it is independent of the input graph. Then for any such mechanism we can construct graphs which yield additive approximation $\Omega(\sqrt{n})$.
	
	First, if there exists any vertex $v \in N$ with $\Pr{v \in S}> 1/\sqrt{n}$, then consider a nomination profile consisting of vertex $v$ having maximum in-degree $\Delta=n-1$ (i.e., all other vertices are pointing to it), with all other vertices having in-degree either $1$ or $0$. Since $u^*$ belongs to the sample (and, hence, cannot be the winner) with probability at least $1/\sqrt{n}$, the expected degree of the winner is at most $1+(n-1)(1-1/\sqrt{n}) = \Delta - \Theta(\sqrt{n})$.
	
	Otherwise, assume that every vertex $v \in N$ has probability at most $1/\sqrt{n}$ of being selected in the sample set. Consider a nomination profile with a vertex $u^* \in N$ having maximum degree $\Delta=\sqrt{n}/2$ and all other vertices having in-degree either $0$ or $1$. Consider a vertex $u$ pointing to vertex $u^*$. The probability that $u$ belongs to the sample is at most $1/\sqrt{n}$. Hence, by the union bound, the probability that some of the $\sqrt{n}/2$ vertices pointing to $u^*$ is selected in the sample set is at most $1/2$. Hence, the probability that $u^*$ is returned as the winner is not higher than $1/2$ and the expected in-degree of the winner is at most $1+\sqrt{n}/2 \cdot 1/2 =\Delta - \Theta(\sqrt{n})$.
\end{proof}

\subsection{General Lower Bound}\label{sec:general-lower-bound}

\begin{figure*}[tbh]
	\includegraphics[width=\textwidth]{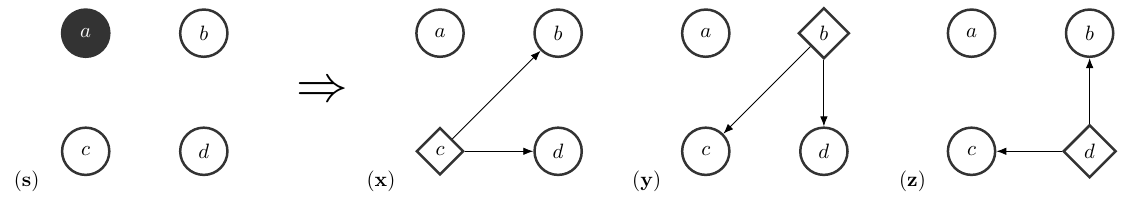}
	\caption{Vertex $a$ is the winner in profile $\mathbf{s}$. This leads to the three profiles $\x$,
		$\y$ and $\z$, where each diamond-shaped vertex cannot be the winner. }
	\label{fig:4GLB:0}
\end{figure*}

	\begin{figure*}[tb]
	\hspace*{\fill}%
	\includegraphics[width=\textwidth]{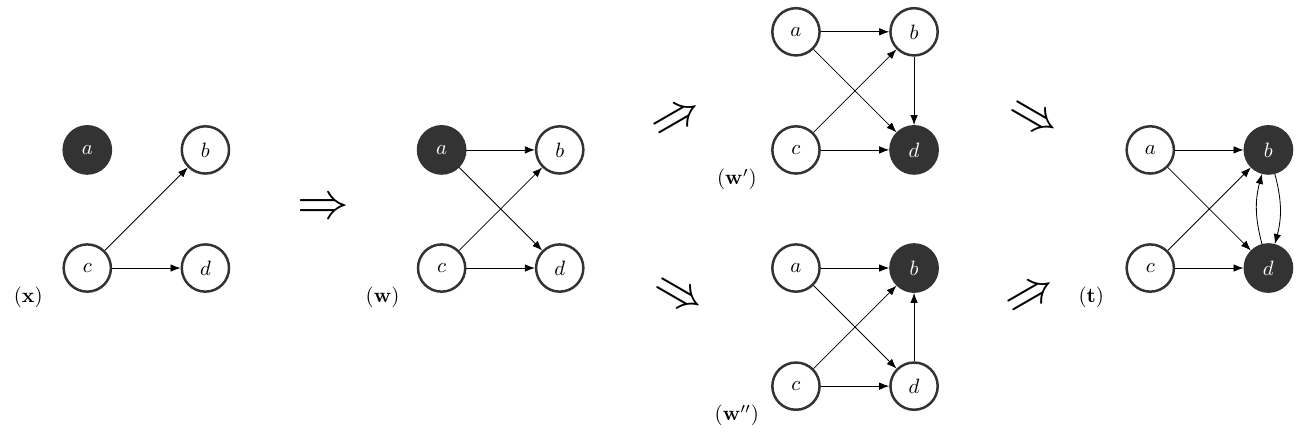}
	\caption{Case 1 in the proof of Theorem~\ref{thm:lower-general}. We assume that $a$ is the winner in profile $\x$. When $a$ adds votes, the profile  $\mathbf{w}$ is created.  
		Then, in the upper profile $\mathbf{w'}$, vertex $d$ must win (for an additive approximation guarantee strictly less than $3$) and in the lower profile $\mathbf{w''}$, vertex $b$ must win. This however leads to the final profile $\mathbf{t}$, where both $d$ and $b$ must win, due to impartiality ---a contradiction.}
	\label{fig:4GLB:no_a}
\end{figure*}

Our last result is a lower bound for all
deterministic impartial mechanisms in the most general model of
Alon et al.~\cite{alon11}, where each agent can nominate multiple other agents or even abstain. We remark that our current proof applies to mechanisms that always select a winner.


\begin{theorem}\label{thm:lower-general}
	There is no impartial deterministic $\alpha$-additive mechanism for $\alpha\leq 2$.
\end{theorem}

\begin{proof}
	Let $f$ be a deterministic impartial mechanism and, for the sake of contradiction, assume that it
	achieves additive approximation at most equal to $2$. We will show that there is a profile with four
	vertices (denoted by $a,b,c$ and $d$), in which the winner has
	in-degree $0$, while the maximum in-degree is $3$, which leads
	to a contradiction.
	
	We first consider the profile with no edges, say $\s$, and let us assume, without loss of generality,
	that the winner is $a$ (see Figure \ref{fig:4GLB:0}).
	Now consider the three profiles $\x$, $\y$ and $\z$ produced when each
	of the other three vertices $c$, $b$ and $d$ vote for the other two of
	them, respectively ($c$ votes for $d$ and $b$, $b$ votes for $c$ and
	$d$, and $d$ votes for $b$ and $c$, as shown in Figure
	\ref{fig:4GLB:0}). In all these profiles, the voter (the vertex which changes its outgoing edges, compared to profile $\s$) cannot be the
	winner since this would break impartiality. Focus for example on the
	profile $\x$. Since $c$ cannot be the winner, it must be either $a$, $b$ or
	$d$. There are essentially two cases, which we
	treat separately.

	\paragraph*{Case 1: $a$ is the winner for at least one of  $\,\x,\y,\z.\,$ } Consider the profile $\x$, where vertex $c$ votes for both $b$ and $d$ and
	assume that $a=f(\x)$ (see Figure~\ref{fig:4GLB:no_a}). We let $a$
	vote for both $b$ and $d$, to get the profile
	$\mathbf{w}=(\{(a,b)(a,d)\},\x_{-a})$. Impartiality implies that $a=f(\mathbf{w})$.

	On the one hand, if $b$ votes for $d$ (profile
	$\mathbf{w}'=(\{(b,d)\},\mathbf{w}_{-b})$), impartiality implies that
	$b\neq f(\mathbf{w}')$ and approximation allows only $f(\mathbf{w}')=d$. 
	On the other hand, if $d$ votes for $b$ (profile
	$\mathbf{w}''=(\{(d,b)\},\mathbf{w}_{-d})$), by similar arguments we have
	$f(\mathbf{w}'')=b$ (see Figure \ref{fig:4GLB:no_a}). Now, consider
	the profile $\ttt$ where both $b$ and $d$ vote for each other,
	i.e., $\ttt=(\{(d,b)\},\mathbf{w'}_{-d})$ and, at the same time,
	$\ttt=(\{(b,d)\},\mathbf{w''}_{-b})$.  Impartiality (applied to
	$\mathbf{w}'$ and $\mathbf{w}''$, respectively) implies that both $b$
	and $d$ must be winners which is absurd and leads to a contradiction.
	Similar arguments would apply for the other cases, establishing that
	$a$ cannot be the winner in any of the profiles $\x, \y$ or $\z$.

		\begin{figure*}[h!bt]
		\centering
		\subfloat[Same winner in profiles $\x$ and $\y$.]{\label{fig:4GLB:BigTable:case21}
			\centering
			\includegraphics[scale=.5]{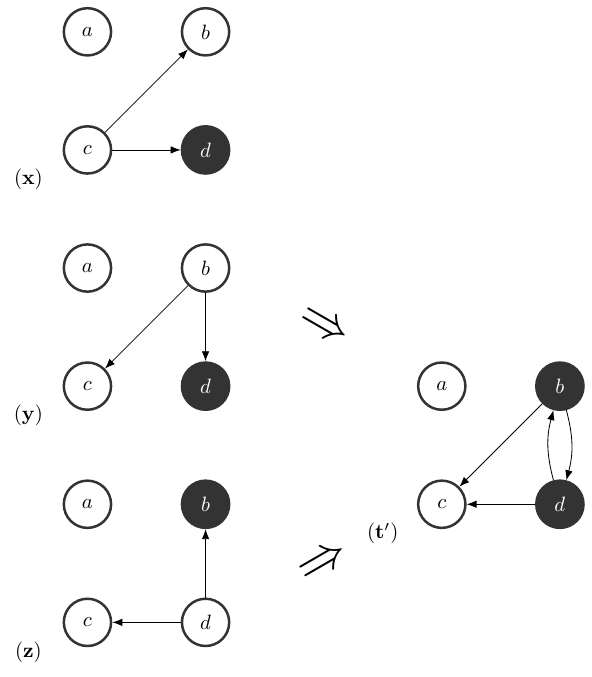}}\hfill
		\subfloat[{Profiles $\x,\y,\z$ have different winners.}]{ \label{fig:4GLB:BigTable:case22}
			\centering
			\includegraphics[scale=.5]{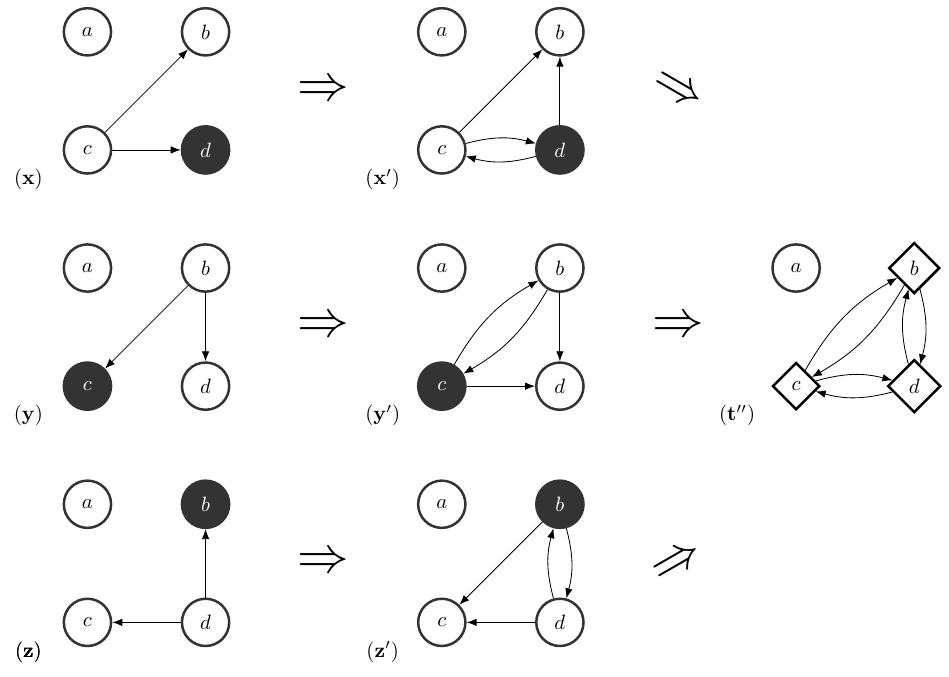}}
		\caption{Case 2 in the proof of Theorem~\ref{thm:lower-general}. Vertex $a$ is not the winner in any of the  profiles $\x, \y, \z$. In figure \ref{fig:4GLB:BigTable:case21}, we assume that two of the winners in $\x,\y,\z$ are the same, and this leads to a mechanism with two winners ---a contradiction. In figure \ref{fig:4GLB:BigTable:case22}, we assume that no two profiles among $\x,\y,\z$ have the same winner. Impartiality implies that in the rightmost profile $\mathbf{t}''$ where three vertices have in-degree $2$, the winner must be $a$, with in-degree $0$. If $a$ votes towards any other vertex, impartiality demands that $a$ remains the winner, with an in-degree of $0$, while the maximum in-degree is $3$ ---a contradiction.}
		\label{fig:4GLB:BigTable}
	\end{figure*}
	
	\paragraph*{Case 2:  $a$ is not the winner for any $\x, \y, \z.\,$ }
	
	In this case, due to impartiality, only vertices with in-degree $1$
	are possible winners. Hence, we are left only with two sub-cases;
	either two of these profiles share the same winner or all of them  have a different winner.
	
	In the first sub-case, consider (without loss of generality) the scenario where $f(\x)=f(\y)$. Impartiality, plus the fact that $a$ is not the winner in $\x$, imply that $f(\x)=d$.
	Assume that $f(\z)=b$ (illustrated in Figure~\ref{fig:4GLB:BigTable:case21}). The alternative case $f(\z)=c$ follows through similar arguments. In profile $\y$, we let $d$ add $2$ votes and create profile $\mathbf{t'}=(\{(d,b)(d,c)\},\y_{-d})$. By impartiality, 
	$f(\mathbf{t'})=d$. In profile $\z$, we let $b$ add $2$ votes and create again the profile $(\{(b,c)(b,d)\},\z_{-b})=\mathbf{t'}$.
	Note that these graphs are, indeed, the same. By impartiality $f(\mathbf{t'})=b$, hence the impartial mechanism $f$ at profile $\mathbf{t'}$
	must award two vertices, a contradiction. Similar arguments hold in all the cases where two of the profiles $\x,\y,\z$ share the same winner.
		
	We are left now with the case where all these profiles, $\x$, $\y$ and
	$\z$ have different winners, where none of them is $a$. There are $2$ possible such scenarios: 
	$f(\x)=d $, $f(\y)=c$ and $f(\z)=b$ (see Figure~\ref{fig:4GLB:BigTable:case22}), or $f(\x)=b $, $f(\y)=d$ and $f(\z)=c$. Consider the first one (similar arguments hold also for the second). From these profiles $\x$, $\y$ and $\z$ we reach the profiles $\x'=(\{(d,b)(d,c)\},\x_{-d})$,
	$\y'=(\{(c,b)(c,d)\},\y_{-c})$ and $\z'=(\{(b,c)(b,d)\},\z_{-d})$, by letting the respective winners to add edges. Because of impartiality, all the winners are preserved, i.e., $f(\x)=f(\x')$, $f(\y)=f(\y')$ and $f(\z)=f(\z')$. Let us now focus on profile $\x'$. By letting vertex $b$ add edges $(b,c)$ and $(b,d)$, we create the profile $\mathbf{t}''=(\{(b,c)(b,d)\},\x'_{-b})$: a directed clique on the vertices $b,c$ and $d$ and the vertex $a$ with no incoming nor outgoing edges (see Figure~\ref{fig:4GLB:BigTable:case22}). Focusing now on profile $\y'$, we reach the profile $(\{(d,b)(d,c)\},\y'_{-d})=\mathbf{t}''$ by a deviation of vertex $d$; the same profile as before. In a similar fashion, on profile $\z'$ we reach the profile $(\{(c,b)(c,d)\},\z'_{-c})=\mathbf{t}''$ by a deviation of $c$. By impartiality,  $f(\mathbf{t}'') \notin \{b,c,d\}$,
	which implies that $f(\mathbf{t}'')=a$. Now, if $a$ votes for at least any other vertex,
	impartiality implies that $a$ must remain the winner, while the
	nominees of $a$ will have in-degree $3$, contradicting the
	approximation guarantee of $f$.
\end{proof}


%
%

\bibliographystyle{plain}      
\bibliography{references}   

%
%

\end{document}